\documentclass[aip,jmp,preprint]{revtex4-1}

\usepackage{amsmath} 
\usepackage{amssymb}
\usepackage{amstext}
\usepackage{amsthm}
\usepackage{amsfonts}
\usepackage{mathtools}
\usepackage{enumerate}  
\usepackage{dsfont}
\usepackage[utf8]{inputenc}

\usepackage[pdftex]{graphicx} 
\usepackage{hyperref}

\theoremstyle{plain}

\newtheorem{theorem}{Theorem}
\newtheorem*{theorem*}{Theorem}
\newtheorem{lemma}[theorem]{Lemma}
\newtheorem{corollary}[theorem]{Corollary}

\newcommand{\C}{\mathds{C}}
\newcommand{\R}{\mathds{R}}
\newcommand{\Z}{\mathds{Z}}

\newcommand{\ii}{\mathrm{i}}
\newcommand{\e}{\mathrm{e}}

\newcommand{\ee}{\mathbf{e}}

\begin{document}

\renewcommand{\labelenumi}{\roman{enumi}.}

\title{Pretty Good State Transfer in Qubit Chains - The Heisenberg Hamiltonian}

\author{Leonardo Banchi}
\affiliation{Department of Physics and Astronomy, University College London, \\ London, WC1E 6BT, United Kingdom}

\author{Gabriel Coutinho} 
\email{gabriel@dcc.ufmg.br}
\affiliation{Departamento de Ciência da Computação, ICEx-UFMG, \\ Belo Horizonte, MG, 31270-901, Brazil}

\author{Chris Godsil}
\affiliation{Department of Combinatorics \& Optimization, University of Waterloo, \\ Waterloo, ON, N2L3G1, Canada} 

\author{Simone Severini} 
\affiliation{Department of Computer Science, University College London, \\ London, WC1E 6BT, United Kingdom}

\date{  \today  }

\begin{abstract}
Pretty good state transfer in networks of qubits occurs when a continuous-time quantum walk allows the transmission of a qubit state from one node of the network to another, with fidelity arbitrarily close to 1. We prove that in a Heisenberg chain with $n$ qubits there is pretty good state transfer between the nodes at the $j$-th and $(n-j+1)$-th position if $n$ is a power of 2. Moreover, this condition is also necessary for $j=1$. We obtain this result by applying a theorem due to Kronecker about Diophantine approximations, together with techniques from algebraic graph theory. 
\end{abstract}

\maketitle

\section{Introduction}

Long-distance quantum communication, for example over several kilometers, 
typically uses photonic systems. 
On the other hand, given the difficulty of engineering interactions 
between photons, several 
promising candidates for  quantum hardware are based on 
{\it quantum networks} of localized qubits \cite{ladd2010quantum}, which are
easier to manipulate. 
In typical quantum algorithms 
during the computation the states of different qubits  have to be transferred between various registers, namely between different nodes of the network.  However, 
when qubits are localized, their physical movement may be either impossible,
by construction, or energy inefficient. 
A viable solution to this problem is to exploit the coherent dynamics of the 
quantum network, namely a continuous-time quantum walk,
to transfer the quantum states between different nodes 
\cite{BoseQuantumComPaths}. This approach for short-distance ({\it in-chip})
communication has attracted much attention 
\cite{nikolopoulos2014quantum} because it minimizes the use of external control
and also avoids the complex interface between localized and moving 
particles. However, in a generic quantum network the resulting coherent dynamics 
is very complicated and the transmission between two nodes is normally inefficient. 
Therefore, much effort has been devoted to understand what are the best strategies,
or the best networks, to achieve either {\it perfect state transfer }
between distant nodes 
\cite{KayReviewPST,GodsilPerfectStateTransfer12,BoseCasaccinoManciniSeverini,burgarth2005conclusive}, or 
{\it pretty good state transfer} \cite{VinetZhedanovAlmost,GodsilKirklandSeveriniSmithPGST,banchi2013ballistic,apollaro201299,lorenzo2013quantum,omar2014pretty,BurgarthPhDThesis,ref1,ref2} where the 
transmission quality is {\it almost} perfect.
One-dimensional systems, namely chains of qubits, are perhaps the most natural 
candidate for transmission as they resemble a quantum wire or {\it data-bus}. 

Most of the literature on perfect or pretty good state transfer has considered 
chains interacting with the XY Hamiltonian (also called XX Hamiltonian in 
the condensed matter community\cite{mikesca}). 
This kind of chains offers several mathematical 
simplifications, e.g. the Hamiltonian in the single-particle subspace is equivalent 
to the adjacency matrix of the corresponding graph and also 
the many-particle problem is exactly solvable 
\cite{albanese2004mirror}. 
However, in solid state systems such as  quantum dots 
\cite{hanson2007spins} or dopants in silicon \cite{kane1998silicon} and in 
current optical lattice experiments \cite{fukuhara2013quantum}, 
which are some of the most promising quantum devices, the 
natural interaction is the Heisenberg (XYZ) Hamiltonian, sometimes also called 
XXX Hamiltonian\cite{mikesca}. 
For engineered chains in the single-particle subspace the distinction between 
XY and XYZ Hamiltonians can be simulated also with local external fields. 
However, in unmodulated (i.e. non engineered) systems, which are easier to 
create experimentally, such a difference is fundamental.  
Motivated by this, in this paper we focus on unmodulated qubit networks described 
by the XYZ Hamiltonian and we find a full characterization of the chains
admitting pretty good state transfer, namely we prove the following result:

\begin{theorem} \label{thm:lapla}
Pretty good state transfer occurs between the extremal vertices of Heisenberg chain of $n$ qubits if and only if $n$ is a power of $2$. Moreover, in these cases, pretty good state transfer occurs between vertices at the $j$th and $(n+1-j)$th position for all $j = 1,...,n$. 
\end{theorem}

The remainder of the paper is organized as follows. In Section \ref{sec:notation}
we introduce the required notation. Section \ref{sec:tec} describes the tools 
needed for the proof of our main result. Specifically Theorem \ref{thm:char}
generalizes for every symmetric algebraic  matrix a result by 
Vinet and Zhedanov\cite{VinetZhedanovAlmost} on certain XY spin chains. 
Finally Section \ref{sec:main} proves our main result. 

\section{Notation and definitions}\label{sec:notation}
We consider a graph $G=(V,E)$ with a set of vertices $V(G)=\{1,\dots,n\}$ and a set 
of edges $E(G)$ that describe the physical pairwise couplings between two vertices. 
We denote $A(G)$ the adjacency matrix with elements $[A(G)]_{ij}=1$, 
if $(i,j)\in E(G)$, 
and $[A(G)]_{ij}=0$ otherwise. For a generic graph structure the Heisenberg (XYZ) 
Hamiltonian is defined by 
\begin{equation}
   \mathcal H_{\rm XYZ}(G) = \frac12\sum_{i\neq j} A(G)_{ij} \left(
     X_i X_j + Y_i Y_j + Z_i Z_j\right),
  \label{heinsenberg}
\end{equation}
where $X_i$, $Y_i$, $Z_i$ are the Pauli matrices acting on the $i$-th vertex. 
On the other hand, the XY Hamiltonian is 
\begin{equation}
   \mathcal H_{\rm XY}(G) = \frac12\sum_{i\neq j} A(G)_{ij} \left(
     X_i X_j + Y_i Y_j \right).
  \label{xyham}
\end{equation}
Both $\mathcal H_{\rm XY}(G)$ and 
$\mathcal H_{\rm XYZ}(G)$ act on the Hilbert space $(\mathbb{C}^2)^{\otimes n}$. 
We call $\{\vert 0\rangle, \vert 1\rangle\}$ the basis of the Pauli matrices on
each vertex and we define the single-particle subspace as the Hilbert space generated 
by the vectors $X_i \vert 0\rangle^{\otimes n} = \vert 0\dots 010\dots 0\rangle
\in(\mathbb{C}^2)^{\otimes n} $, for $i=1,\dots,n$, 
where the $\vert 1\rangle$ state is in the $i$-th position. 
Within this single-particle subspace, the
above Hamiltonians can be written as (see Ref.\cite{BoseCasaccinoManciniSeverini})
\begin{align}
   \mathcal H^{(1)}_{\rm XYZ}(G) &= |E(G)| \openone - 2 L(G), \\ 
   \mathcal H^{(1)}_{\rm XY}(G) &= 2 A(G), 
\end{align}
where the subscript $(1)$ refers to the single-particle subspace, 
$L(G) = \Delta(G)-A(G)$ is the Laplacian of the graph and 
$\Delta(G)$ is the diagonal matrix whose diagonal $i$-th entry is the degree $d(i)$ of 
vertex $i$, namely the number of edges incident with $i$. 
For simplicity, we avoid the use of the explicit notation 
$\mathcal H^{(1)}_{\rm XYZ}(G)$,  $\mathcal H^{(1)}_{\rm XY}(G)$,  and 
we simply call $A(G)$ and $L(G)$ as the XY and XYZ Hamiltonians, 
as they are equivalent to equations (\eqref{heinsenberg}) and (\eqref{xyham}) 
in the single-particle subspace up to a trivial rescaling and shift. 

We now introduce the concept of perfect and pretty good state transfer. 
Given $M$ a symmetric matrix whose columns are indexed by the set of vertices $V$, 
we say that \textit{perfect state transfer} occurs between vertices $a$ and $b$ of 
$M$ if there is a $\tau \in \R^+$ such that
\[|\exp(\ii \tau M)_{a,b}| = 1.\]
This framework generalizes the concept of state transfer in the quantum walk 
of XY and XYZ Hamiltonians in the single-excitation subspace where 
$M$ is respectively chosen as $A(G)$ or $L(G)$. 
If it is clear from the context which $M$ we are dealing with, we use the notation $\exp(\ii t M) = U(t)$.

We relax the definition of perfect state transfer to an $\epsilon$-version. We say that $M$ admits \textit{pretty good state transfer} (also known as almost perfect state transfer) between vertices $a$ and $b$ if, for any $\epsilon > 0$, there is a time $\tau > 0$ such that 
\begin{align} | U(\tau) _{a,b} | > 1 - \epsilon,  \label{eq1} \end{align}
If $\ee_a$ and $\ee_b$ are the characteristic vectors of columns $a$ and $b$, equation (\ref{eq1}) is equivalent to the existence of a $\lambda \in \C$ of absolute value equal to $1$ such that
\[ || U(\tau) \ee_a - \lambda \ee_b || < \epsilon.\]
Finally, for shortness, when $\epsilon$ is not relevant, we abbreviate this equation to
\[  U(\tau)\ee_a \approx \lambda \ee_b.\]

Given a real symmetric matrix $M$ with spectral decomposition
\[M = \sum_{r = 0}^d \theta_r E_r,\]
we say that $a$ and $b$ are \textit{strongly cospectral} if $E_r \ee_a = \pm E_r \ee_b$ for all $r$. This nomenclature is inspired by the following fact. We say that vertices $a$ and $b$ are \textit{cospectral} if the matrix obtained from $M$ upon removing row and column indexed by $a$ has the same spectrum as when we remove row and column indexed by $b$. An equivalent formulation is that $(E_r)_{a,a} = (E_r)_{b,b}$ for all $r$, therefore every pair of strongly cospectral vertices is cospectral, as one would expect. If $M$ is either the adjacency or the Laplacian matrix of a graph, cospectral vertices have necessarily the same number of neighbours. Moreover, in the adjacency case, $a$ and $b$ are cospectral if and only if, for all $k \in \Z$, the number of walks of length $k$ that start and end in $a$ is the same as the number for $b$ (see Ref.\cite[Section 2.5]{CoutinhoPhD} for proofs and references of these facts). There are cases in which cospectral vertices are not strongly cospectral, and in fact we do not know any combinatorial characterization of this property. Finally, it is worth mentioning that if all eigenvalues are simple, both properties are equivalent, and that if $M$ is a tridiagonal matrix (thus encoding the adjacency of a linear chain), then strong cospectrality is equivalent to the property of mirror-symmetry of the weights.

We also define the \textit{eigenvalue support} of $a$ as the set of eigenvalues $\theta_r$ such that $E_r \ee_a \neq 0$.

\section{Technical preliminaries}\label{sec:tec}

Godsil et al.~\cite{GodsilKirklandSeveriniSmithPGST} determined when a linear chain with unmodulated spins admits pretty good state transfer between the end vertices according to the XY-Hamiltonian. Subsequently, Vinet and Zhedanov \cite{VinetZhedanovAlmost} applied a theorem due to Kronecker on chains with non-unitary weights, providing new examples of pretty good state transfer in the XY-Hamiltonian model.

We point out that the application of Kronecker's theory to characterize pretty good state transfer in Heisenberg chains of prime length was previously considered\cite{BurgarthPhDThesis}. Unfortunately the fact that the eigenvalues of these Hamiltonians are not linearly independent over the rationals went unnoticed, leading to the false conclusion that pretty good state transfer would occur in all chains of prime length. As an example, note that the probability of state transfer at time $t$ between the extreme vertices in a chain of $3$ qubits is equal to ${\frac{1}{9}(1-\cos(t))^2(5+4 \cos(t))}$, which is at most $3/4$ for all $t$.

In this work, we apply a more descriptive version of Kronecker's theorem, which deals with the case where the eigenvalues are not necessarily linearly independent over the rationals. As a result, we fully characterize linear chains with unmodulated spins admitting pretty good state transfer according to the Heisenberg Hamiltonian.

\begin{theorem} \label{thm:char}
Let $a$ and $b$ be columns of a symmetric algebraic matrix $M$. Then pretty good state transfer occurs between $a$ and $b$ if and only if both conditions below are satisfied.
\begin{enumerate}[(i)]
\item Columns $a$ and $b$ are strongly cospectral. In this case, let $\theta_0 , ... , \theta_d$ be the eigenvalues in their support, and for $r = 0,...,d$, let $\sigma_r$ be defined as $0$ if the projections onto $E_r$ are equal, and $1$ if they have opposite signs.
\item If there is a set of integers $\ell_0,....,\ell_d$ such that
\[ \sum_{r = 0}^d \ell_r \theta_r = 0 \quad \text{and} \quad \sum_{r = 0}^d \ell_r \sigma_r \text{ is odd,}\]
then
\[\sum_{r = 0}^d \ell_r \neq 0.\]
\end{enumerate}
\end{theorem}

Condition (i) is known to be necessary for perfect state transfer (see for instance Ref.\cite{KayPerfectcommunquantumnetworks}). In the Lemma below, a slightly modified argument works to show that it is also necessary for pretty good state transfer. 

\begin{lemma}
If pretty good state transfer occurs between $a$ and $b$, then they are strongly cospectral vertices.
\end{lemma}
\begin{proof}
From the spectral decomposition, we have
\[U(t) = \sum_{r=0}^d \e^{\ii t \theta_r} E_r,\]
thus
\[|U(t)_{a,b}| \leq  \sum_{r = 0}^d |(E_r)_{a,b}|.\]
Now $\sum E_r = I$, and, by Cauchy-Schwartz, 
\[(E_r)_{a,a} \geq |(E_r)_{a,b}|.\]
Thus
\[\sum_{r = 0}^d |(E_r)_{a,b}| \leq 1\]
and equality holds if and only if, for all $r$,
\[(E_r)_{a,a} = |(E_r)_{a,b}|,\]
or equivalently, $a$ and $b$ are strongly cospectral. Therefore if $a$ and $b$ are not strongly cospectral, then threre is $\epsilon$ such that 
\[\sum_{r = 0}^d |(E_r)_{a,b}| < 1 - \epsilon,\]
and so, for all $t$, it follows that 
\[|U(t)_{a,b}| \leq 1 - \epsilon,\]
hence pretty good state transfer does not occur.
\end{proof}

We will make use of the following result due to Kronecker.

\begin{theorem}[Kronecker, see for instance Ref.\cite{AlmostPeriodicFunctionsBook}, Chapter 3] \label{thm:kro}
Let $\theta_0,...,\theta_d$ and $\zeta_0,...,\zeta_d$ be arbitrary real numbers. For an arbitrarily small $\epsilon$, the system of inequalities
\[| \theta_r y - \zeta_r | < \epsilon \pmod {2\pi} , \quad (r = 0,...,d),\]
admits a solution for $y$ if and only if, for integers $\ell_0,...,\ell_d$, if
\[\ell_0 \theta_0 + ... + \ell_d \theta_d = 0,\]
then
\[\ell_0 \zeta_0 + ... + \ell_d \zeta_d \equiv 0 \pmod {2\pi}.\]
\end{theorem}

Now we prove our characterization.

\begin{proof}[Proof of Theorem \ref{thm:char}]

Observe that
\[  U(\tau)\ee_a \approx \lambda \ee_b \]
is equivalent to, for all $r$,
\[   \e^{\ii \theta_r \tau } E_r \ee_a \approx \lambda  E_r \ee_b,\]
which in turn, when $\lambda = \e^{\ii \delta}$, is equivalent to, for all $r$ such that $E_r \ee_u \neq 0$,
\begin{align} \theta_r \tau \approx \delta + q_r \pi,  \label{eq2} \end{align}
where $q_r \in \Z$ is even if and only if $E_r \ee_u = E_r \ee_v$, and odd if and only if $E_r \ee_u = - E_r \ee_v$.


A solution to equation (\ref{eq2}) is equivalent to a solution as described in Theorem \ref{thm:kro} with
\[ y = \tau \quad \text{and} \quad \zeta_r = \delta + \sigma_r \pi, \]
where $\sigma_r = 0$ if $E_r \ee_a = E_r \ee_b$, and $\sigma_r = 1$ if $E_r \ee_a = -E_r \ee_b$.

Now, a particular set of integers $\ell_0,...,\ell_d$ satisfies
\[\ell_0 \zeta_0 + ... + \ell_d \zeta_d \equiv 0 \pmod {2\pi}\]
if and only if there is a $\delta$ such that
\[\ell_0 (\delta + \sigma_0 \pi) + ... + \ell_d (\delta + \sigma_d \pi)  \equiv 0 \pmod {2\pi}\]
which in turn is equivalent to
\begin{align}\delta \left( \sum_{r = 0}^d \ell_r \right) + \pi \left( \sum_{r = 0}^d \sigma_r \ell_r \right) \equiv 0 \pmod {2\pi}.\label{eq3} \end{align}
A solution $\delta$ to this equation exists if and only if, if $\sum \sigma_r \ell_r$ is odd, then $\sum \ell_r$ is non-zero, precisely as stated in condition (ii). This shows immediately that condition (ii) is necessary.

Note however that in order to apply Kronecker's Theorem and show sufficiency, we still need to show that the $\delta$ chosen to solve ($\ref{eq3}$) should also work for any other set of integers $\ell_0',...,\ell_d'$ satisfying $\sum \ell_r' \theta_r = 0$. The result follows easily if, for all such sets, $\sum \sigma_r \ell_r'$ is even, in which case $\delta = 0$ is a solution. So assume $\ell_0,...,\ell_d$ are such that $\sum \sigma_r \ell_r$ is odd, and so $\delta \sum \ell_r$ is also odd, and let $\ell_0',...,\ell_d'$ be such that $\sum \ell_r'\theta_r = 0$. Let $\alpha = \sum \ell_r$ and $\beta = \sum \ell_r'$, and define integers $\gamma_r = (\delta/\pi)(\beta \ell_r - \alpha \ell_r')$. If $\sum \sigma_r \ell_r'$ and $(\delta/\pi) \sum \ell_r'$ have opposing parities, then $\sum \gamma_r \theta_r = 0$, $\sum \sigma_r \gamma_r$ is odd, and $\sum \gamma_r = 0$. This contradicts condition (ii). Therefore $\sum \sigma_r \ell_r'$ and $(\delta/\pi)\sum \ell_r'$ have the same parity, and thus $\delta$ and $\ell_0',...,\ell_d'$ satisfy ($\ref{eq3}$) as well.
\end{proof}

This next corollary is notably useful to study the Laplacian matrix.
\begin{corollary}
Assume $M$ is the Laplacian matrix of a graph with strongly cospectral vertices $a$ and $b$. Say $\theta_0 = 0$, and so $\sigma_0 = 0$. Say the other eigenvalues in their support are $\theta_1 , ... ,\theta_d$, and have $\sigma_1,...,\sigma_d$ defined as before. Then pretty good state transfer occurs between $a$ and $b$ if and only if whenever there are integers $\ell_1,....,\ell_d$ such that
\[ \sum_{r = 1}^d \ell_r \theta_r = 0,\]
then
\[\sum_{r = 1}^d \sigma_r \ell_r \quad \text{is even.}\]
Moreover, in this case, the complex phase with which pretty good state transfer occurs will be equal to $1$.
\end{corollary}
\begin{proof}
Make $\theta_0 = 0$. Then given $\ell_0,...,\ell_d$,
\[\sum_{r = 0}^d \ell_r \theta_r = 0 \quad \iff \quad \sum_{r = 1}^d \ell_r \theta_r = 0.\]
Hence the choice of $\ell_0$ is arbitrary, and thus can always be made such that
\[\sum_{r = 0}^{d} \ell_r = 0.\]
Thus, in order for pretty good state transfer to occur, $\sum_{r = 1}^d \sigma_r \ell_r$ can never be odd, and if it is even in all cases, condition (ii) of Theorem \ref{thm:char} is vacuously satisfied. Moreover, in this case, as the choice $\ell_0$ is arbitrary and hence can also be made in a way that $\sum_{r = 0}^{d} \ell_r$ is odd, $\delta$ must be an even multiple of $\pi$, therefore $\lambda = \e^{\ii \delta} = 1$.
\end{proof}

\subsection{The spectrum of Heisenberg chains} \label{sec:graphspectra}

We refer the reader to Brouwer and Haemers \cite{BrouwerHaemers} for the result below. Let $P_n$ denote the path on $n$ vertices. Recall that $L(X)$ denotes the Laplacian matrix of the graph $X$.

\begin{itemize}
\item The eigenvalues of $L(P_n)$ are $0$ with the all 1s eigenvector, and $2 + 2 \cos(\pi r / n)$, $r = 1,...,n-1$. If $\beta_k = \sin( k \pi r /n )$, its corresponding eigenvector is \[( \beta_1 ,\ (-1)^{1} (\beta_1 + \beta_2) ,\  (\beta_2 + \beta_3) ,... ,\  (-1)^{n} (\beta_{n-2} + \beta_{n-1}) , \  (-1)^{n+1}\beta_{n-1}  ) .\]

\end{itemize}

\section{Main result}\label{sec:main}

We are ready to prove Theorem \ref{thm:lapla}.

\smallskip
\noindent {\bf Theorem 1 (restated)} {\it
Pretty good state transfer occurs on $L(P_n)$ between the extremal vertices if and only if $n$ is a power of $2$. Moreover, in these cases, pretty good state transfer occurs between vertices at the $j$th and $(n+1-j)$th position for all $j = 1,...,n$. 
}
\smallskip

\begin{proof}
Suppose the spectral decomposition of $L(P_n)$ is given by
\[L(P_n) = \sum_{r = 0}^{n-1} \lambda_r E_r.\]
Let $R$ be the anti-diagonal matrix of order $n$. It is a straightforward consequence of the spectrum of $P_n$ described in Section \ref{sec:graphspectra} that
\[\sum_{r = 0}^{n-1} (-1)^r E_r = R.\]
This readily implies that vertices at positions $j$ and $(n+1-j)$ are strongly cospectral for $j = 1,...,n$, and hence condition (i) of Theorem \ref{thm:char} is always satisfied, with $\sigma_r = [1+(-1)^{r+1}]/2$.

Let $\zeta_{2n} = \e^{\pi / n}$. Clearly the eigenvalues of $P_n$ can be expressed as
\[\lambda_r = 2 - (\zeta_{2n}^{\phantom{n}r} + \overline{\zeta_{2n}^{\phantom{n}r}} ) = 2 - (\zeta_{2n}^{\phantom{n}r} + \zeta_{2n}^{\phantom{n}2n-r} ) .\]
As a consequence, the eigenvalues belong to the cyclotomic field of $\zeta_{2n}$. Now assume there are integers $\ell_1,...,\ell_{n-1}$ such that
\begin{align}\sum_{r = 1}^{n-1} \ell_r \left( - 2 + (\zeta_{2n}^{\phantom{n}r} + \zeta_{2n}^{\phantom{n}2n-r} ) \right)  = 0. \label{eq4} \end{align}
If $\ell_0 = - \sum_{r=1}^{n-1} \ell_r$, then the cyclotomic polynomial $\Phi_{2n}(x)$ divides
\begin{align} L(x) = 2\ell_0 + \sum_{r = 1}^{n-1} \ell_r x^r   + \sum_{r = n+1}^{2n-1} \ell_{2n-r} x^{r}.\label{eqL(x)} \end{align}
\begin{enumerate}[(i)]
	\item If $n$ is a power of $2$, then $\Phi_{2n}(x) = 1 + x^n$. Performing long division starting from the terms of smaller degree, the general form of an exact quotient of the division of $L(x)$ by $\Phi_{2n}(x)$ is
	\[2 \ell_0 + \sum_{r = 1}^{n-1} \ell_r x^r,\]
	thus the division is exact (and equation (\ref{eq4}) is satisfied) if and only if $\ell_0 = 0$ and $\ell_{s} = \ell_{n-s}$ for all $s = 1,...,n$. As a consequence, whenever (\ref{eq4}) holds, $\sum \ell_{\texttt{odd}}$ is always even, and pretty good state transfer occurs.
	\item If $n$ is an odd prime, then $\Phi_{2n}(x) = 1 - x + x^2 - ... + x^{n-1}$. Performing long division starting from the terms of smaller degree, the general form of an exact quotient of the division of $L(x)$ by $\Phi_{2n}(x)$ is
	\[  2\ell_0 + (2\ell_0 + \ell_1)x + \sum_{r = 2}^{n-1} (\ell_r + \ell_{r-1} ) x^r + \ell_{1}x^n.\]
	This implies that a set of integers $\ell_1,...,\ell_{n-1}$ satisfy equation (\ref{eq4}) if and only if, for all odd $s$ between $1$ and $n-1$,
	\begin{align}\ell_{s} - \ell_{n-s} =  -2\ell_0 = 2 \sum_{r = 1}^{n-1} \ell_r. \label{eq5} \end{align}
	If $n \equiv 3 \pmod 4$, then, for any odd $s$, define $\ell_{s} = n$ and $\ell_{n-s} = -(n-2)$. This provides a solution to ($\ref{eq5}$) such that
	\[\sum_{s \text{ odd}} \ell_{s}  \quad \text{is odd,}\]
	hence pretty good state transfer does not occur in this case. 
	
	If $n \equiv 1 \pmod 4$, make $\ell_1 = -6$, $\ell_{n-1} = 2(n-2)$, and, for any other odd $s$, $\ell_s = -(n+4)$ and $\ell_{n-s} = (n-2)$. Hence
	\[\sum_{r = 1}^{n-1} \ell_r = -(n+1),\]
	and so, for all $s$ odd,
	\[\ell_s - \ell_{n-s} = 2 \sum_{r = 1}^{n-1} \ell_r,\]
	but
	\[\sum_{s \text{ odd}} \ell_s = -(n+4) \left(\frac{n-3}{2}\right) - 6 \quad\text{is odd},\]
	hence pretty good state transfer does not occur in this case.
	
	\item Now assume $n = m k$, where $k > 1$ is odd and $m \geq 2$. We have the identity
	\[1 + 2 \sum_{r = 1}^{(k-1)/2} (-1)^r \cos \left( \frac{\pi r}{k} \right) = 0.\]
	For $q \in \{1,2\}$, we multiply both sides by $\cos (q \pi / n)$ to obtain
	\[\cos \left(\frac{q\pi}{n} \right) + \sum_{r = 1}^{(k-1)/2} (-1)^r \left[ \cos \left( \frac{\pi (mr + q)}{n} \right) + \cos \left( \frac{\pi (mr - q)}{n} \right) \right] = 0.\]
	Now subtract the equation above for $q=1$ from the equation for $q=2$. This leads to
	\[(\lambda_1 - \lambda_2) + \sum_{r=1}^{(k-1)/2} (-1)^r (\lambda_{mr+1} - \lambda_{mr+2}) + \sum_{r=1}^{(k-1)/2} (-1)^r (\lambda_{mr-1} - \lambda_{mr-2}) = 0.\]
	This is an integer combination of the eigenvalues equal to $0$. The coefficients multiplying the odd-indexed eigenvalues appear in pairs, with the exception of the first, which is equal to $1$. Therefore pretty good state transfer does not occur.
\end{enumerate}
\end{proof}

\section{Final remarks}

Theorem \ref{thm:char} provides conditions which are necessary and sufficient for pretty good state transfer to occur, but fails to offer an estimate of how large $t$ should be in order to obtain an approximation in terms of a given $\epsilon$. In Ref.\cite{HemmerMaximonWergeland}, it is shown that in certain dynamical systems the time is proportional to $(1/\epsilon)^n$, but this makes the assumption that the numbers $\theta_0,...,\theta_d$ in the statement of Kronecker's Theorem are linearly independent over the rationals. Note that this is not the case for the examples studied in our paper, and we do not know a reference for these situations.

We point out that we are not necessarily determining all cases in which a qubit chain might admit pretty good state transfer, as we are focusing only on transfer between the end vertices. In fact, there are known examples of perfect state transfer in linear chains with weighted edges between internal nodes while not occurring between the end vertices (see Ref.\cite{KayReviewPST}). 

In an earlier version of our paper, we left as an open question to characterize pretty good state transfer between any pair of vertices in a linear chain with unit weights according to the XY or the XYZ Hamiltonian. For the XY Hamiltonian, some progress was made in Ref.\cite{CoutinhoGuoVBommel}, and the problem was subsequently solved in Ref.\cite{VBommelPGST}. The question remains open for the Laplacian case.

\begin{acknowledgements}
LB acknowledges funding from the European Research Council under the European Union’s Seventh Framework Programme (FP/2007-2013) / ERC Grant Agreement No. 308253.
GC acknowledges the support of grants FAPESP 15/16339-2 and FAPESP 13/03447-6. 
CG acknowledges the support of NSERC Grant RGPIN-9439. 
SS acknowledges the Royal Society and EPSRC.

We all acknowledge the valuable comments of the anonymous referee. We also acknowledge Gabor Lippner, Ada Chan, and Pedro Baptista, whose comments helped us to substantially improve the quality of our paper and fix some mistakes.
\end{acknowledgements}

%

\end{document}